\documentclass[sigconf,authorversion,nonacm,table,prologue]{acmart}
\pdfoutput=1

\usepackage{algorithm2e}
\usepackage{centernot}
\usepackage{comment}
\usepackage{etoolbox}
\usepackage{fancyvrb}
\usepackage[htt]{hyphenat}
\usepackage[utf8]{inputenc}
\usepackage{listings}
\usepackage{alloy-style}
\usepackage{upquote}
\usepackage{url}

\usepackage{tikz}
\usepackage{pgfplots}
\pgfplotsset{compat=1.15}

\makeatletter
\patchcmd{\maketitle}{\@copyrightspace}{}{}{}
\makeatother

\makeatletter
\newcommand{\removelatexerror}{\let\@latex@error\@gobble}
\makeatother

\SetKwFor{For}{for}{}{}

\colorlet{punct}{red!60!black}
\definecolor{delim}{RGB}{20,105,176}
\definecolor{darkblue}{HTML}{43A2CA}
\definecolor{darkgreen}{HTML}{A8DDB5}

\lstdefinelanguage{json}{
    basicstyle=\normalfont\ttfamily\scriptsize,
    showstringspaces=false,
    breaklines=true,
    literate=
      {:}{{{\color{punct}{:}}}}{1}
      {,}{{{\color{punct}{,}}}}{1}
      {\{}{{{\color{delim}{\{}}}}{1}
      {\}}{{{\color{delim}{\}}}}}{1}
      {[}{{{\color{delim}{[}}}}{1}  
      {]}{{{\color{delim}{]}}}}{1}, 
}

\begin{document}

\title{Fast Discovery of Nested Functional and Inclusion Dependencies on JSON Data}

\author{Michael J. Mior}
\email{mmior@cs.rit.edu}
\affiliation{
  \institution{Rochester Institute of Technology}
  \streetaddress{102 Lomb Memorial Drive}
  \city{Rochester}
  \state{NY}
  \country{USA}
  \postcode{14623-5608}
}

\begin{abstract}
Functional and inclusion dependencies are the most widely used classes of data dependencies in data profiling due to their ability to identify relationships in data such as primary and foreign keys.
These relationships are equally important when dealing with nested data formats such as JSON.
However, the definition of functional and inclusion dependencies makes use of a flat, unnested relational model which misses many useful types of dependencies on data which involve nested data models.

In this work, we identify types of dependencies which are not captured by traditional functional and inclusion dependencies but which nevertheless capture meaningful relationships among nested data.
We also demonstrate how algorithms for mining these traditional dependencies can be adapted to also mine nested dependencies.
The first strategy simply flattens the input data and feeds into unmodified existing algorithms.
We present a second strategy which instead adapts the algorithm to efficiently process JSON data as input which in some cases leads to a reduction in runtime by multiple orders of magnitude on real-world datasets.
We further show how these algorithms can be adapted to produce useful results in the presence of a percentage of incomplete or invalid data.
\end{abstract}

\maketitle

\section{Introduction}

Primary and foreign key relationships provide information on datasets that is useful for data scientists when analyzing an unknown set of data.
However, in many cases this metadata is unavailable either because it was never defined by the author of the data or the system being used does not provide support for these definitions.
As an alternative, these keys can be inferred from the data using dependency mining.
Primary keys are often represented by \emph{functional dependencies} while foreign keys are represented by \emph{inclusion dependencies}.
Both of these classes of dependencies define relationships between two sets of attributes.
For example, a functional dependency may indicate that any two tuples with the same value of \texttt{employeeID} must also have the same value for \texttt{firstName} and \texttt{lastName}, ensuring that each ID only associates a single name with an employee.
Inclusion dependencies instead indicate that all values for a set of attributes also exist as values within another set of attributes.
For example, if all values of the \texttt{managerID} attribute also appear as values of the \texttt{employeeID} attribute, we can infer that \texttt{managerID} is likely a foreign key to \texttt{employeeID}.

Unfortunately, these dependencies (and by extension, algorithms for mining them) are not designed to express relationships on nested data.
However, such relationships can prove useful.
Consider an example of Amazon product data in Figure~\ref{fig:amz_data} slightly modified from real data collected from Amazon~\cite{He2016}.
We note that within each document, when a value is contained in the \texttt{buy\_after\_viewing} array, it implies that it must also be contained in the \texttt{also\_viewed} array.
Similarly, the \texttt{also\_viewed} array nested under the \texttt{related} key contains a list of Amazon Standard Identification Numbers (ASINs).
These ASIN values refer to the \texttt{asin} attribute of other documents in the collection.
This allows cross-referencing detailed information on other products viewed by users who have purchased a given product.
Mining these relationships provides significant value in describing the data stored in collections of JSON documents.

\begin{figure*}
    \centering
\begin{lstlisting}[frame=single,language=json]
{
  "asin": "B007IJKOMK",
  "salesRank": {"Music": 513528},
  "related": { "also_viewed": ["B00284G31G", "B001HADE96"], "buy_after_viewing": ["B001HADE96"]},
  "categories": [ ["CDs & Vinyl", "Classical"], ["Musical Instruments", "Instrument Acc."] ]
}
\end{lstlisting}
    \caption{Sample Amazon product data}\label{fig:amz_data}
\end{figure*}

While this is not unique to JSON data, we also identify that is common for data to be incorrect or incomplete.
Such ``dirty'' data can result in semantically meaningful and correct dependencies to be considered invalid.
That is, if our approach to mining validates dependencies against all documents, a single erroneous document in a large collection would render a dependency invalid.
Therefore, in this work, we also consider approaches to dependency mining which can find ``approximate'' dependencies which hold on the majority of a dataset, while permitting a configurable threshold of violations to accommodate data which may be incorrect.
This significantly increases the utility of our algorithms since it is much easier to allow small amounts of erroneous data than it is to identify and correct the errors.

In this work, we extend the notion of functional and inclusion dependencies to allow dependencies between attributes in collections of nested JSON documents.
Specifically, we make the following contributions:

\begin{enumerate}
    \item A definition of nested functional and inclusion dependencies suitable that extends traditional inclusion dependencies for nested JSON data
    \item Efficient algorithms for mining these dependencies on a collection of JSON documents
\item Approximation techniques to allow for discovery of dependencies which are valid on the majority of a dataset, enabling our algorithms to work on dirty data
    \item An evaluation of the effectiveness of this approach in discovering dependencies in real-world datasets
\end{enumerate}

\section{Background}

Our definitions of nested dependencies are extensions of the traditional notions of functional and inclusion dependencies.
We first provide a summary of these traditional definitions.
We then describe JSONPath that we will use to represent the nested attributes on either side of the dependencies.

\subsection{Functional Dependencies}\label{subsec:fds}

Functional dependencies are used to identify a relationship between a set of attributes in a single relation.
For example, a functional dependency $\texttt{employeeID}\rightarrow\texttt{firstName},$\allowbreak$\texttt{lastName}$ indicates that any two tuples which have the same value for \texttt{employeeID} must also have the same values for \texttt{firstName} and \texttt{lastName}.
Commonly, such a functional dependency would be described as \texttt{employeeID} determines \texttt{firstName} and \texttt{lastName}.
In general, functional dependencies are of the form $A\rightarrow B$ where $A$ and $B$ are sets of attributes on some relation.
While functional dependencies do not directly imply uniqueness of any attributes values, they are often used to identify primary keys for a relation.
For a relation $R\left(A\right)$, the functional dependency $K\rightarrow A\setminus K$ indicates that the attribute set $K$ is a primary key for $R$.
Due to this relationship between functional dependencies and primary keys, functional dependencies are commonly used in the context of data profiling for identifying possible primary keys for a relation.

Functional dependencies (FDs) are also useful for identifying and enforcing constraints.
For example, the FD $\texttt{zipCode}\rightarrow\texttt{city}, \texttt{state}$ indicates that the same zip code cannot be assigned to multiple different cities.
Identifying these types of dependencies can help identify and repair errors in dirty data~\cite{Bohannon2005,Kolahi2009}.
For example, suppose we identify that if the majority of records with zip code 14623 have the city and state Rochester, NY.
If a record exists with zip code 14623 with a city and state of Rochester, MN, we can infer that the recorded state of MN is likely incorrect.

When working with JSON documents, functional dependencies are equally useful for identifying such constraints and primary key relationships.
However, as we identify in Section~\ref{sec:nfd}, there are several cases where the traditional definition fails to capture relationships in nested JSON data which is the focus of this work.

\subsection{Inclusion Dependencies}\label{subsec:inds}

Inclusion dependencies have long been used as a manner of identifying possible foreign key relationships in relational databases.
An inclusion dependency $R\left(A\right)\subseteq S\left(B\right)$ indicates that for every tuple $t\in R$, there exists a corresponding tuple in $u\in S$ where these two tuples agree on the sets of attributes $A$ and $B$ (i.e., $t\left[A\right]=u\left[B\right]$).
This property is necessary and sufficient for identifying a valid foreign key relationship and is particularly useful when foreign keys are not maintained for a dataset.

NoSQL databases making use of the JSON data format are among systems that commonly have no explicit representation of foreign keys.
There is rarely an explicit representation of the database schema nor any notion of referential integrity as is the case with all popular document databases such as MongoDB~\cite{MongoDB}, DynamoDB~\cite{DeCandia2007}, and Couchbase~\cite{Couchbase}.
This complicates the task of data exploration.
When processing an unknown dataset, the lack of knowledge of foreign key relationships limits the types of analyses that can be performed.
For example, the analyst may not know what possible joins will produce useful results.

One approach to facilitate this exploration is inclusion dependency mining.
As we discuss further in Section~\ref{sec:related_work}, existing approaches to inclusion dependency mining are not designed to discover inclusion dependencies on nested JSON data.
When all documents consist of JSON objects with a single level of nesting, existing approaches to inclusion dependency mining can be easily adapted. We simply treat all JSON documents as tuples in a relation defined by the possible attributes present across all documents.

However, as mentioned previously, this fails to capture inclusion dependencies that arise within nested objects or arrays.
Consider again the example from Figure~\ref{fig:amz_data}.
One simple approach to adapting existing algorithms would be to ``unroll'' such nested data into relational format as in Table~\ref{tbl:amz_unrolled}.
However, we find that unrolling a collection of $\sim$100MB of JSON documents similar to the one presented in Figure~\ref{fig:amz_data} results in records totalling over 1.7GB\@!
As we show in Section~\ref{sec:eval}, this technique, which we term \emph{static unrolling}, results in significant overhead since it requires the mining algorithm to process a large number of tuples where much of the data is redundant.
Later in Section~\ref{sec:nind}, we describe an adaption of inclusion dependencies to handle nested structures.

\begin{table*}
\centering
{\footnotesize
\rowcolors{2}{white}{gray!25}
\begin{tabular}{lllllll}
asin       & \begin{tabular}{@{}l@{}}salesRank \\ KEY\end{tabular} & \begin{tabular}{@{}l@{}}salesRank \\ VALUE\end{tabular} & \begin{tabular}{@{}l@{}}related \\ also\_viewed\end{tabular} & \begin{tabular}{@{}l@{}}related \\ buy\_after\_viewing\end{tabular} & categories\_0       & categories\_1          \\
\hline{}B007IJKOMK & Music          & 513528           & B00284G31G & B001HADE96 & CDs \& Vinyl        & Classical              \\
B007IJKOMK & Music          & 513528           & B001HADE96 & B001HADE96          & CDs \& Vinyl        & Classical              \\
B007IJKOMK & Music          & 513528           & B00284G31G & B001HADE96 & Musical Instruments & Instrument Acc. \\
B007IJKOMK & Music          & 513528           & B001HADE96 & B001HADE96          & Musical Instruments & Instrument Acc.
\end{tabular}}
    \caption{Unrolled Amazon product data}\label{tbl:amz_unrolled}
\end{table*}

\subsection{JSONPath}

As discussed previously, our focus in this work is on JSON data.
JavaScript Object Notation (JSON) is a common format used to represent documents in NoSQL databases.
JSON consists of strings, numbers, booleans, arrays, and objects.
Both arrays and objects can be nested to an arbitrary depth.
A simplified version of the grammar for JSON is given in Figure~\ref{fig:json_grammar} with the document in Figure~\ref{fig:amz_data} being an example instance.

\begin{figure}
\begin{verbatim}
pair: STRING ':' value;
obj: '{' pair (',' pair)* '}' | '{}';
array: '[' value (',' value)* ']' | '[]';
value: STRING | NUMBER | obj | array
      | 'true' | 'false' | 'null';
\end{verbatim}
\caption{Simplified JSON grammar~\protect\cite{Parr2013}}~\label{fig:json_grammar}
\end{figure}

To denote nested fields within each document, we make use of  JSONPath~\cite{Goessner2007}.
JSONPath is a language similar to XPath~\cite{Clark1999} for XML that identifies specific nodes within the nested document tree.
When defining paths for use in our nested dependencies, we consider the root of each path (denoted \texttt{\$} in JSONPath) to be a document in our collection.
Paths are simply represented using dot notation while we can consider all elements of an array using \texttt{[*]}.
For example, the path $\texttt{\$.related.also\_viewed[*]}$ refers to all elements of the array under the key \texttt{related.also\_viewed}.

\section{Nested Inclusion Dependencies}\label{sec:nind}

As demonstrated in the previous section, the standard definition of inclusion dependencies for the relational model fails to capture the entirety of inclusion relationships in JSON data.
Specifically, treating arrays as atomic misses many relationships within JSON data.

For example, consider the traditional notion of an inclusion dependency augmented by using attributes produced from evaluating the following two JSONPath expressions: \texttt{\$.related.also\_viewed[*]} and  \texttt{\$.asin} on the data in Figure~\ref{fig:amz_data}.
Since evaluating the path \texttt{\$.related.also\_viewed[*]} produces an array and \texttt{\$.asin} produces an atomic value, this inclusion dependency would never hold.
(That is, an array would never appear as an element in a set of atomic values.)
However, this misses an important opportunity to express a meaningful property of the data.

As discussed previously, the array value which is located at the path \texttt{\$.related.also\_viewed} \emph{and} the string value at the path \texttt{\$.asin} both contain values from the same domain.
The path \texttt{\$.asin} produces the ASINs of all products while we can think of the path \texttt{\$.related.also\_viewed} as a foreign key connecting a produt to other products commonly purchased together.

To capture these types of relationships, we extend the traditional notion of inclusion dependencies to \emph{nested inclusion dependencies} (NINDs).
Instead of using attributes from a relation on either side of the dependency, we use a sequence of JSONPath expressions.

In this work, we consider nested inclusion dependencies within a single collection of documents, so there is no need to indicate a document collection on either side of the dependency.
Since it is common for JSON objects to contain redundant and duplicated data, nested inclusion dependencies within a single collection of documents are still useful.
However, our approach can be easily extended to incorporate dependencies involving multiple collections of documents (i.e. with attributes from different collections on each side).

\subsection{Definition}

Valid JSONPath expressions for use in nested inclusion dependencies result in a value which is either atomic or an array.
In the case where a path expression produces an atomic value, the semantics of nested inclusion dependencies are identical to traditional inclusion dependencies.
When an expression produces an array, we have four cases to consider: either an expression on the left results in an array, an expression on the right, neither, or both.

For simplicity, we first consider a nested inclusion dependency with only a single path on each side when describing the semantics which we denote $\mathcal{D}: P_1\subseteq P_2$ where $\mathcal{D}$ is a collection of documents and $P_1$ and $P_2$ are JSONPath expressions.
For a document $d\in\mathcal{D}$, we use the notation $d\left[P\right]$ to refer to the value stored at path $P$.
Similarly, we will use $\mathcal{D}\left[P\right]$ to refer to the set of all values produced when evaluating $d\left[P\right]$ for all $d$ in $\mathcal{D}$.

We say that a value $d\left[P_1\right]$ is \emph{fully included} in $\mathcal{D}\left[P_2\right]$, which we denote $d\left[P_1\right]\Subset\mathcal{D}\left[P_2\right]$ if any of the following conditions are met:
\begin{enumerate}
  \item $v=d\left[P_1\right]$ is an atomic value and $v\in\mathcal{D}\left[P_2\right]$,
  \item $v=d\left[P_1\right]$ is an atomic value and there is a document $d'\in\mathcal{D}$ where $d'\left[P_2\right]$ is an array and $v\in d'\left[P_2\right]$, or
  \item $a=d\left[P_1\right]$ is an array and all values $v\in a$ in the array exist either as atomic values or within arrays in $\mathcal{D}\left[P_2\right]$ according to 1 and 2 above.
\end{enumerate}

A nested inclusion dependency $\mathcal{D}:P_1\subseteq P_2$ then holds iff $d\left[P_1\right]\Subset\mathcal{D}\left[P_2\right]\forall d\in\mathcal{D}$.
That is, all values in the set of values at the path $P_1$ across all documents in $\mathcal{D}$ are fully included in the set of values at $P_2$ in the collection $\mathcal{D}$.

\subsection{Examples}

We now present several examples of nested inclusion dependencies including both those which mirror standard inclusion dependencies as well as those with nested values at the path on the left, right, or both sides.
For simplicity, we will consider only three documents which are shown in Figure~\ref{fig:example_docs}.

\begin{figure*}
    \centering
    \begin{minipage}[t]{0.24\textwidth}
        \begin{lstlisting}[language=json]
        {
          "id": 2",
          "parent": 5,
          "rel": [3,5],
          "main": "A",
          "types": ["A"]
        }
        \end{lstlisting}
    \end{minipage}
    \begin{minipage}[t]{0.22\textwidth}
        \begin{lstlisting}[language=json]
        {
          "id": 3",
          "rel": [2],
          "main": "D",
          "types": ["D"]
        }
        \end{lstlisting}
    \end{minipage}
    \begin{minipage}[t]{0.26\textwidth}
        \begin{lstlisting}[language=json]
        {
          "id": 4",
          "parent": 2,
          "rel": [2],
          "main": "C",
          "types": ["A","C"]
        }
        \end{lstlisting}
    \end{minipage}
    \begin{minipage}[t]{0.26\textwidth}
        \begin{lstlisting}[language=json]
        {
          "id": 5",
          "related": [2],
          "main": "B",
          "types": ["B","C"]
        }
        \end{lstlisting}
    \end{minipage}
    \caption{Example JSON documents}
    \label{fig:example_docs}
\end{figure*}

First, we have the inclusion dependency $\mathcal{D}:\texttt{\$.parent}\subseteq\texttt{\$.id}$ which is simply between the atomic values at these paths and equivalent to a traditional inclusion dependency we may have expressed if this data were represented using the relational model.
We also have the inclusion dependency and $\mathcal{D}:\texttt{\$.rel[*]}\subseteq\texttt{\$.id}$.
This is because all values at the paths \texttt{\$.parent} and \texttt{\$.id} appear in some values of the path \texttt{\$.rel[*]}.
These dependencies could not be expressed using the traditional notion of inclusion dependencies since the flat relational model used does not permit arrays of values.
Nevertheless, these dependencies express the useful fact that the paths \texttt{\$.parent} and \texttt{\$.rel[*]} can be considered as references to values at the path \texttt{\$.id}.

As discussed previously, we can also express a dependencies between Amazon product identifiers in the dataset in Figure~\ref{fig:amz_data}.
For example, $\mathcal{D}: \texttt{\$.related.also\_viewed}\subseteq\texttt{\$.asin}$.
This expresses that all items at the path \texttt{\$.related.also\_viewed} also appear as values in some document at the path \texttt{\$.asin}.
That is, all products a customer has viewed are represented as references to the identifier of the product in another document.

\subsection{Inference Rules}

Inference rules similar to those for relational inclusion dependencies apply equally to nested inclusion dependencies. 
These inference rules are important since the mining algorithms we explore depend on their validity.
We provide brief proofs below.

\begin{lemma}[Reflexivity of NINDs]
The nested inclusion dependency $\mathcal{D}:\mathbf{P_1}\subseteq \mathbf{P_2}$ always holds if $P_1=P_2$.
\end{lemma}

\begin{proof}
For the dependency to hold, we must have for all $d\in\mathcal{D}$, $d\left[P_1\right]\Subset \mathcal{D}\left[P_2\right]$.
For each $d\in\mathcal{D}$, $d\left[P_1\right]$ is either an atomic value or an array.
Since paths are the same on both sides of both sides of the dependency, if $d\left[P_1\right]$ atomic value, we have $d\left[P_1\right]=d\left[P_2\right]$ and the first condition for full inclusion is satisfied.
If $d\left[P_1\right]$ is an array, then each value of $d\left[P_1\right]$ is also contained in the array $d\left[P_2\right]$ since the two are equivalent.
Therefore, $d\left[P_1\right]$ is fully included in $\mathcal{D}\left[P_2\right]$ and the dependency holds.
\end{proof}

\begin{lemma}[Transitivity of NINDs]
Given two inclusion dependencies $\mathcal{D}:P_1\subseteq P_2$ and $\mathcal{D}:P_2\subseteq P_3$, the dependency $\mathcal{D}:P_1\subseteq P_3$ also holds.
\end{lemma}

\begin{proof}
The validity of this rule depends on the transitivity of the full inclusion relationship ($\Subset)$.
If $\mathcal{D}:P_1\subseteq P_2$ holds, then all values at the path $P_1$ are fully included in some value at the path $P_2$.
Similarly, if $\mathcal{D}:P_2\subseteq P_3$ holds, then all values at the path $P_2$ are fully included in some value at the path $P_3$.
Since our full inclusion relationship amounts to containment testing, it is transitive which implies that all values at the path $P_1$ must be fully contained in some value at the path $P_3$.
Therefore, the dependency $\mathcal{D}:P_1\subseteq P_3$ must hold.
\end{proof}

\section{Nested Functional Dependencies}\label{sec:nfd}

We now provide a similar definition and set of inference rules for nested functional dependencies.
Like nested inclusion dependencies, these dependencies generalize functional dependencies in the presence of nested data.

\subsection{Definition}

A \emph{nested functional dependency} (NFD) takes similar form to a traditional functional dependency.
The nested functional dependency $\mathcal{D}:\mathbf{P}_1\rightarrow \mathbf{P}_2$ indicates that the values at the set of paths $\mathbf{P}_1$ in $\mathcal{D}$ collectively determine the values at the paths $\mathbf{P}_2$.
This notion of collective determinacy is equivalent to the notion of determinacy in traditional functional dependencies when the JSON document has no array values.
We start by considering functional dependencies with a single attribute on each side.
For simplicity, we treat any atomic values as arrays with a single element, e.g., a value $3$  will be interpreted as $\left[3\right]$.
We say that the values $\mathcal{D}\left[P_{1,1}\cdots P_{1,m}\right]$ collectively determine $\mathcal{D}\left[P_{2,1}\cdots P_{2,n}\right]$ iff for all documents $d_1,d_2\in \mathcal{D}$, one of the following conditions holds:

\begin{enumerate}
  \item $d_1\left[P_{1,i}\right]\cap d_2\left[P_{1,i}\right]=\emptyset$ for some $i\in\{1\ldots m\}$ or
  \item $d_1\left[P_{2,j}\right]\cap d_2\left[P_{2,j}\right]\neq \emptyset$ for all $j\in\{1\ldots n\}$.
\end{enumerate}
That is, whenever the values at all paths in $\textbf{P}_1$ for a document $d_1$ overlap with the values in a second document $d_2$, the values at each path in $\textbf{P}_2$ for document $d_1$ must also overlap with the values for each path in $d_2$.

Note that since we treat atomic values as single-element sets, intersection is equivalent to equality for atomic values, making nested functional dependencies a superset of traditional functional dependencies.

\subsection{Examples}

Considering again the example documents in Figure~\ref{fig:example_docs}, one possible example of a nested functional dependency is $\texttt{\$.id}\rightarrow\texttt{\$.parent}$.
Since the values at \texttt{\$.id} serve as keys for the document, this is no different from the traditional notion of functional dependencies other than the data being represented as JSON.
We also have the dependency $\texttt{\$.id}\rightarrow\texttt{\$.rel[*]}$.
Again, this has a rough equivalent to the traditional notion of functional dependencies.
The primary difference is that the values on the right side of this dependency are arrays.

Consider again the example of Amazon product data in Figure~\ref{fig:amz_data}.
Each element of array at the path \texttt{\$.categories} represents a sequence of categories and subcategories.
We can represent the relationship between the top-level category and the first subcategory with the nested functional dependency $\texttt{\$.categories[*][1]}\rightarrow\texttt{\$.categories[*][0]}$.
This indicates that whenever two documents have first subcategories which are the same, they must also have top-level category in common.

\subsection{Inference Rules}\label{subsec:ind_rules}

Armstrong's axioms are a standard set of inference rules applied to functional dependencies including reflexivity, augmentation, and transitivity.
Similar axioms hold for our definition of nested functional dependencies which we define and prove below.
As with nested inclusion dependencies, the validity of these axioms is important to ensure the correctness of our mining algorithms.
Note that while these axioms hold for all nested functional dependencies, we also found model checking useful in the construction of our proofs.
We used the Alloy model checker~\cite{Jackson2002} to construct a model of nested functional dependencies which we provide in Appendix~\ref{sec:alloy}.

\begin{lemma}[Reflexivity of NFDs]
The nested functional dependency $\mathcal{D}:\mathbf{P_1}\rightarrow \mathbf{P_2}$ always holds if $\mathbf{P}_2\subseteq\mathbf{P}_1$.
\end{lemma}

\begin{proof}
If $\textbf{P}_2\subseteq\mathbf{P}_1$, then by definition, values at each path in $\mathbf{P}_2$ will have a non-empty intersection with the values at paths in $\mathbf{P}_1$, so $\mathcal{D}: \mathbf{P_1}\rightarrow\mathbf{P}_2$ holds.
\end{proof}

\begin{lemma}[Augmentation of NFDs]
If $\mathbf{P}_1$, $\mathbf{P}_2$, and $\mathbf{P}_3$ are sets of paths and $\mathcal{D}: \mathbf{P}_1\rightarrow\mathbf{P}_2$, then $\mathcal{D}:\mathbf{P}_1\mathbf{P}_3\rightarrow\mathbf{P}_2\mathbf{P}_3$ also holds.
\end{lemma}

\begin{proof}
For $\mathcal{D}:\mathbf{P}_1\mathbf{P}_3\rightarrow\mathbf{P}_2\mathbf{P}_3$ to hold, we must have all paths in $\mathbf{P}_2\mathbf{P}_3$ overlapping when paths in $\mathbf{P}_1\mathbf{P}_3$ overlap.
We know that values at the paths in $\mathbf{P}_2$ overlap when values at the paths in $\mathbf{P}_1$ overlap since $\mathcal{D}:\mathbf{P}_1\rightarrow\mathbf{P}_2$ holds.
Values at the paths in $\mathbf{P}_3$ always overlap with values at the paths in $\mathbf{P}_3$ since they are equivalent.
Therefore values at the paths in $\mathbf{P}_1\mathbf{P}_3$ will overlap if values in $\mathbf{P}_2\mathbf{P}_3$ overlap and $\mathcal{D}:\mathbf{P}_1\mathbf{P}_3\rightarrow\mathbf{P}_2\mathbf{P}_3$ holds.
\end{proof}

\begin{lemma}[Transitivity of NFDs]
If the nested functional dependencies $\mathcal{D}:\mathbf{P}_1\rightarrow \mathbf{P}_2$ and $\mathcal{D}:\mathbf{P}_2\rightarrow\mathbf{P}_3$ hold, then the dependency $\mathcal{D}:\mathbf{P}_1\rightarrow\mathbf{P}_3$ also holds.
\end{lemma}

\begin{proof}
If $\mathcal{D}:\mathbf{P}_1\rightarrow\mathbf{P}_3$ is to hold, we must show that either $d_1\left[P_{1,i}\right]\cap d_2\left[P_{1,i}\right]$ is empty or $d_1\left[P_{3,i}\right]\cap d_2\left[P_{3,i}\right]$ is \emph{not} empty $\forall i\in\{1\ldots n\}$.
We consider the case where $d_1\left[P_{1,i}\right]\cap d_2\left[P_{1,i}\right]$ is not empty.
If this is the case, then we must have that $d_1\left[P_{2,i}\right]\cap d_2\left[P_{2,i}\right]$ is not empty since $\mathcal{D}:\mathbf{P}_1\rightarrow\mathbf{P}_2$ holds.
As a result, we know that $d_1\left[P_{3,i}\right]\cap d_2\left[P_{3,i}\right]$ must hold since $\mathcal{D}:\mathbf{P}_2\rightarrow\mathbf{P}_3$.
Therefore, $\mathcal{D}:\mathbf{P}_1\rightarrow\mathbf{P}_3$.
\end{proof}

\section{Discovery Algorithms}

We implement our techniques on top of the SPIDER~\cite{Bauckmann2007} algorithm and the mining algorithm by De Marchi~\cite{DeMarchi2002} et al. for inclusion dependency mining and the TANE~\cite{Huhtala1999} and FDep~\cite{Flach1999} algorithms for functional dependency mining.
(For convenience, we refer to the algorithm proposed by De Marchi et al. in the remainder of this work as \emph{DeMarchi}.)
These algorithms were selected for their ease of implementation and generally strong performance.
We expect it will be possible to adapt our techniques to a wide variety of other mining algorithms.

We can use both of these algorithms with minimal changes by modifying the input via \emph{static unrolling} which we explain in Section~\ref{subsec:static}.
That is, we take the nested JSON data and turn it into a flat relational table.
However, this has the disadvantage of significantly expanding the data that must be processed since we must materialize the cross-product of all nested objects.

We note that all the algorithms we consider consist of two phases.
First, metadata is built based on reading the input.
Second, this metadata is processed to determine valid dependencies.
Section~\ref{subsec:dynamic} introduces \emph{dynamic unrolling}, that instead provides the values of nested objects incrementally without materialization.
Sections~\ref{subsec:spider} to~\ref{subsec:fdep} then describe our modifications of SPIDER, DeMarchi, TANE, and FDep respectively to work with dynamic unrolling. 
In Section~\ref{subsec:approx}, we also show how we can adapt these algorithms to generate approximate dependencies which tolerate missing or incorrect values.

\subsection{Static Unrolling}\label{subsec:static}

One way to use existing relational algorithms on existing data is to flatten nested structures into a relational schema.
As shown in Section~\ref{subsec:inds}, this can significantly grow the number of tuples which need to be processed since we need to include documents for all combinations of nested arrays.

We define static unrolling recursively.
There are three cases to consider: objects, arrays, and atomic values.
For objects, we recursively unroll each key and take the cross-product of all generated documents into a single document with all the paths.
For arrays, we simply unroll each element and collect all the generated documents.
Finally, for atomic values, we produce an array with a single document with the given key and value.
The documents produced from this process are free of nesting and all have the same keys (assigning null values to keys which don't exist in a particular document).
This allows us to treat them as relational tuples.
The full algorithm is given in Algorithm~\ref{alg:static}.

\begin{algorithm}[ht]
\SetKwFunction{Flatten}{Flatten}
\SetKwFunction{Unroll}{Unroll}
\SetKwProg{myproc}{Procedure}{}{}
{\myproc{\Flatten{path, value}}{
\uIf{value is empty}{\tcp{Produce a single empty value}\Return{[\{path: ''\}]}}
\uElseIf{value is object}{
  \tcp{Take the cross-product of documents for each key}
  $docs$ $\leftarrow$ [] \\\For{(key, objValue) in value}{$path$ $\leftarrow$ $path$ + '.' + $key$\\$new$ $\leftarrow$ \Flatten{path, objValue}\\$docs$ $\leftarrow$ $docs\times new$}\Return{docs}
}
\uElseIf{value is array}{
  \tcp{Combine flattened documents}
  $path$ $\leftarrow$ $path$ + '[*]'\\
  $docs$ $\leftarrow$ []\\\For{arrayValue in value}{\For{doc in \Flatten{path, arrayValue}}{$docs$.append($doc$)}}\Return{docs}
}
\uElse{\Return{[\{path: value\}]}}
}}
\bigskip
{\myproc{\Unroll{input}}{
$rows$ $\leftarrow$ [] \\
\For{jsonObject in input}{
  \For{row in \Flatten{null, jsonObject}}{$rows$.append($row$)}
}\Return{rows}}}
\caption{Static unrolling}\label{alg:static}
\end{algorithm}

The number of rows scales with the size of nested arrays.
Furthermore, when using static unrolling, the runtime is not proportional to the number of documents, but the number of unrolled rows.
As we show later in Section~\ref{sec:eval}, this has a significant impact on performance.
In the following section, we show a generic approach for adapting existing algorithms without the overhead of static unrolling.

\subsection{Dynamic Unrolling}\label{subsec:dynamic}

As discussed previously, the first step in many mining algorithms is to collect metadata from the input. 
This metadata is traditionally collected row by row from relational datasets.
When performing this step on nested JSON data, flattening is required.
This means a single JSON document can yield a large number of rows as described in the previous section.
Looking at the metadata construction step, we make the important observation that the metadata can instead be constructed simply by providing a series of attribute-value pairs as opposed to an entire document.

To this end, we take the approach of \emph{dynamic} unrolling.
Dynamic unrolling replaces the metadata construction step of the mining algorithm by dynamic traversal of the JSON documents instead of flattening.
This significantly reduces the total size of the data which needs to be processed.
This algorithm makes use of two abstract functions: \texttt{InitializeMetadata} and \texttt{UpdateMetadata}.
This makes the dynamic unrolling approach adaptable to a number of different dependency mining algorithms.

\begin{algorithm}[ht]
\SetKwFunction{Collect}{Collect}
\SetKwFunction{UpdateMetadata}{UpdateMetadata}
\SetKwFunction{InitializeMetadata}{InitializeMetadata}
\SetKwProg{myproc}{Procedure}{}{}
{\myproc{\Collect{id, meta, path, value}}{
\uIf{value is a object}{\For{(key, objValue) in value}{\Collect{id, meta, path + \textquotesingle.\textquotesingle + key, objValue}}}
\uElseIf{value is an array}{\For{arrayValue in value}{\Collect{id, meta, path + \textquotesingle[*]\textquotesingle, arrayValue}}}
\ElseIf{value is not empty}{\UpdateMetadata{id, meta, path, value}}
}}
\bigskip
meta $\leftarrow$ \InitializeMetadata{} \\
$id$ $\leftarrow$ 1 \\
\For{jsonObject in input}{\For{(key, value) in jsonObject}{\Collect{id, meta, \textquotesingle\$.\textquotesingle + key, value}} $id$ $\leftarrow$ $id$ + 1}
\caption{Dynamic unrolling}\label{alg:dynamic}
\end{algorithm}

The \texttt{Collect} function simply recursively explores JSONPaths in the document, calling the \texttt{UpdateMetadata} function each time a new value at a path is discovered.
Each document is also assigned a unique ID which is useful when associating paths with which document they are derived from.
Depending on the nested structure of the data, dynamic unrolling can result in the algorithm processing an order of magnitude or more less data.
The second step differs for each mining algorithm, but generally requires minimal modifications and validates the dependencies according to the collected metadata which has a similar format to the metadata collected in the original implementations.

\subsection{Dynamic SPIDER}\label{subsec:spider}

\begin{algorithm}[ht]
\SetKwFunction{InitializeSPIDERMetadata}{InitializeSPIDERMetadata}
\SetKwFunction{UpdateSPIDERMetadata}{UpdateSPIDERMetadata}
\SetKwProg{myproc}{Procedure}{}{}
{\myproc{\InitializeSPIDERMetadata{}}{
\tcp{Start with an empty object}
\Return \{\}
}}
\bigskip
{\myproc{\UpdateSPIDERMetadata{id, meta, path, value}}{
\uIf{path in meta}{
\tcp{Add the value to the set of values at this path}
$meta$[$path$].add($value$)
}\uElse{
\tcp{Start a new set of values at this path}
$meta$[$path$] = \{$value$\}
}
}}
\caption{Metadata updates for SPIDER}\label{alg:spider-meta}
\end{algorithm}

SPIDER operates by starting with the assumption that all inclusion dependencies hold.
Then for each value $v$, if attribute $A$ has value $v$, the list of valid inclusion dependencies is updated to remove any inclusion dependency $A\subseteq B$ if $B$ does not also contain value $v$.
The first step in the process is producing a sorted list of all unique values for all documents at a given path.
The metadata update functions given in Algorithm~\ref{alg:spider-meta} collect the necessary information.

In our case, instead of collecting a single atomic value per column for each tuple in a relation, we collect all possible values at each path in a document.
Since the document a value was collected from is not relevant to the mining process, it is perfectly acceptable for a document to contribute multiple values.
After the metadata collection step using dynamic unrolling completes, we can proceed with the original SPIDER algorithm.
Note that in the case of JSON documents with no nesting, we will find exactly the same set of dependencies as if we converted each document to the relational model and used the original SPIDER implementation.

\subsection{Dynamic DeMarchi}\label{subsec:demarchi}

The DeMarchi algorithm is similar to SPIDER, but tracks metadata in the opposite direction.
Instead of collecting a set of values for a given path, we collect all paths for a given value.
The metadata updates required are shown in Figure~\ref{alg:demarchi-meta}.

As with SPIDER, DeMarchi starts by assuming that all inclusion dependencies are valid.
Dependencies are invalidated by checking each path recorded for a value and examining all dependencies where this path is on the left-hand side.
We then restrict the valid right-hand sides of each dependency to only the paths which also have this value.

\begin{algorithm}[ht]
\SetKwFunction{InitializeDeMarchiMetadata}{InitializeDeMarchiMetadata}
\SetKwFunction{UpdateDeMarchiMetadata}{UpdateDeMarchiMetadata}
\SetKwProg{myproc}{Procedure}{}{}
{\myproc{\InitializeDeMarchiMetadata{}}{
\tcp{Start with an empty object}
\Return \{\}
}}
\bigskip
{\myproc{\UpdateDeMarchiMetadata{id, meta, path, value}}{
\uIf{value in meta}{
\tcp{Add the path to the set for this value}
$meta$[$value$].add($path$)
}\uElse{
\tcp{Start a new set of paths for this value}
$meta$[$value$] = \{$path$\}
}
}}
\caption{Metadata updates for DeMarchi}\label{alg:demarchi-meta}
\end{algorithm}


\subsection{Dynamic TANE}\label{subsec:tane}

Our first algorithm for mining functional dependencies is modeled after TANE~\cite{Huhtala1999}.
TANE is based around two key techniques: partitioning input tuples based on column values and a levelwise approach for dependency testing that enables efficient pruning.
Since our data involves multi-valued attributes, we cannot partition documents into distinct sets as is required.
TANE uses partitions of input tuples to efficiently test whether a particular dependency holds by comparing the partitions based on different column values.

To enable a similar approach while allowing multi-valued attributes, we instead construct an adjacency matrix of documents for each possible paths.
Two documents are considered adjacent for a particular attribute if the set of values of that attribute for each document intersect.
Since this adjacency matrix is symmetric and sparse, we use a compressed bitmap representation of the matrix (specifically Roaring bitmaps~\cite{Lemire2016}).
This has several advantages.
First, when we need to check if attributes at two paths intersect, we can simply take the intersection of the two bitmaps.
Second, we can validate a dependency by checking if the bitmap representation of the left-hand side is a subset of the representation of the right-hand side.
This ensures that all documents which intersect on attribute values at paths on the left-hand side also intersect on attribute values at paths on the right-hand side.
This precisely matches our definition of nested functional dependencies.

The main downside of this approach is that it loses the scalability of TANE's approach since the initial construction of these bitmaps requires iterating through all pairs of matching tuples for each path in order to set the appropriate bit in the bitmap.
We provide the metadata updates used in dynamic TANE in Figure~\ref{alg:tane-meta}.

\begin{algorithm}[ht]
\SetKwFunction{InitializeTANEMetadata}{InitializeTANEMetadata}
\SetKwFunction{UpdateTANEMetadata}{UpdateTANEMetadata}
\SetKwProg{myproc}{Procedure}{}{}
{\myproc{\InitializeTANEMetadata{}}{
\tcp{Start with empty objects}
\Return \{$value\_indexes$: \{\}, $load\_partitions$: \{\}\}
}}
\bigskip
{\myproc{\UpdateTANEMetadata{id, meta, path, value}}{
\tcp{Add to the set of discovered paths}
\uIf{value not in meta.value\_indexes}{
  $value\_index$ = $meta$.$value\_indexes$[$value$]
}
$meta$.$value\_indexes$[$value$] $\leftarrow$ |$meta$.$value\_indexes$| \\
\uIf{path not in meta.load\_partitions}{
\tcp{Start a new set of partitions}
$meta$.$load\_partitions$[$path$] = \{\}
}
\uIf{$value\_index$ not in $meta$.$load\_partitions$[$path$]}{
\tcp{Start a new partition for this value}
$meta$.$load\_partitions$[$path$][$value\_index$] = $\emptyset$
}
\tcp{Add the document to the partition}
$meta$.$load\_partitions$.[$path$][$value\_index$].add($id$)
}}
\caption{Metadata updates for TANE}\label{alg:tane-meta}
\end{algorithm}

After the dynamic unrolling step is complete, we proceed as in the original TANE algorithm.
The exception is that instead of using partition refinement to check the validity of dependencies, we use our bitmap intersection test as defined above.
Compared to TANE using static unrolling, our algorithm will produce identical results for collections of documents without arrays.
However, since our definition of nested functional dependencies differs from the traditional functional dependencies mined using TANE, our algorithm produces a different set of dependencies in the presence of documents with array values.

\subsection{Dynamic FDep}\label{subsec:fdep}

FDep~\cite{Flach1999} is described as a ``machine learning`` approach to functional dependency mining.
There are two primary stages to the original algorithm.
First, pairs of tuples are compared to see which column values are equal and which are not.
Based on this information, the \emph{negative cover} is constructed.
The negative cover is a set of dependencies that do not hold based on observed violations.
From this negative cover, FDep then constructs the \emph{positive cover}, which is the remaining dependencies which do hold.

\begin{algorithm}[ht]
\SetKwFunction{InitializeFDepMetadata}{InitializeFDepMetadata}
\SetKwFunction{UpdateFDepMetadata}{UpdateFDepMetadata}
\SetKwProg{myproc}{Procedure}{}{}
{\myproc{\InitializeFDepMetadata{}}{
\tcp{Start with an empty object}
\Return \{$docs$: \{\}, $paths$: $\emptyset$\}
}}
\bigskip
{\myproc{\UpdateFDepMetadata{id, meta, path, value}}{
\tcp{Add to the set of discovered paths}
$meta$.$paths$.add($path$) \\
\uIf{id in meta.docs}{
\tcp{Start a new set of values at this path}
$meta$.$docs$[$id$] = \{$path$: \{$value$\}\}
}\uElse{
\tcp{Add to the set of values at this path}
$meta$.$docs$[$id$][$path$].add($value$)
}
}}
\caption{Metadata updates for FDep}\label{alg:fdep-meta}
\end{algorithm}

For our modifications to FDep to support nested functional dependencies, we only need to modify the first step of negative cover construction.
Once the negative cover construction is complete, the rest of the algorithm can proceed as normal.
The metadata we collect via dynamic unrolling is shown in Algorithm~\ref{alg:fdep-meta}.
We simply gather a set of paths as well as all of the documents with their nested structures flattened to a single path string.
Note that this does not result in the same overhead as static unrolling since we are not generating new values to process, but simply restructuring the existing data within each document.

From here, we can proceed with the negative cover construction according to the original FDep algorithm using the documents in $meta$.$docs$ as input to the algorithm.
However, where FDep checks value equality to determine violated functional dependencies, we simply change the check to set intersection on values at a given path.
This means we can still use the original FDep algorithm mostly unchanged, but we gain the ability to discover nested functional dependencies.
As with our dynamic unrolling version of TANE, this version of the FDep algorithm will discover a different set of dependencies when provided a collection of documents with arrays.

\section{Approximate Mining}\label{subsec:approx}

When dealing with JSON documents, it is common to have missing fields, incorrect data types, and changing document structures.
The result is that even dependencies which are semantically valid may not hold across all documents.
One of goals with this work is to allow the discovery of dependencies even in the presence of these anomalies.
Instead of showing only exact dependencies, we would instead like to provide a set of possible dependencies ranked along with some measure of their ``strength''.

For example, an inclusion dependency $A\subseteq B$ may have some values for attribute $A$ that do not appear in attribute $B$.
For a functional dependency, we may have a small fraction of documents which introduce violations.
That is, a dependency $A\rightarrow B$ may not hold on a collection of documents $\mathcal{D}$, but does hold when some small subset $\mathbf{d}\subset\mathcal{D}$ is removed.
Instead of only reporting dependencies that hold exactly, we allow \emph{approximate dependencies} where we relax our definitions to allow for some missing or incorrect values.

We note that there are two conflicting definitions which are used for approximate dependency mining.
The first refers to algorithms which use sampling or other techniques to reduce runtime at the expense of possibly missing some valid dependencies.
The second definition refers to algorithms which find all valid dependencies but also include dependencies with some specified fraction of violating tuples.
We use the latter definition in this work and we adapt each of the algorithms in the previous sections to mine approximate dependencies.

\begin{figure}[ht]
    \centering
\begin{lstlisting}[language=json]
1: {"a": ["X"], "b": ["X", "Y"]}

2: {"a": ["X", "Y"], "b": ["X", "Z"]}

3: {"a": ["X"], "b": ["Y"]}

4: {"a": ["X"], "b": ["Z"]}
\end{lstlisting}
    \caption{Sample documents to demonstrate approximate mining}\label{fig:approx_data}
\end{figure}

\subsection{Inclusion Dependencies}

As described above, our modifications to the SPIDER algorithm for mining inclusion dependencies first collect sorted values for each path.
It then starts by assuming all dependencies are true and iterates through the sorted values to invalidate dependencies whenever values are found not to be included in another path.
In the original SPIDER algorithm, a dependency $B\subset A$ is invalidated as soon as a value for $B$ is found which is not contained in $A$.
Instead, we count the values of $B$ which are contained in $A$ and allow the user to specify a threshold of included values which must be met in order for the dependency can be considered valid.
This same approach works for our modifications to the DeMarchi algorithm.

Consider the example documents in Figure~\ref{fig:approx_data}.
and the inclusion dependency $\texttt{\$.b[*]}\subset\texttt{\$.a[*]}$.
We can see that the values at key \texttt{b} are \texttt{"X"}, \texttt{"Y"}, and \texttt{"Z"} and the values at key \texttt{a} are \texttt{"X"} and \texttt{"Y"}.
Document 2 contains the value \texttt{"Z"} for key \texttt{b} which does not appear in any document at the key \texttt{a}.
After these values are collected, the comparison step of our inclusion dependency mining algorithms would show that two thirds of the values at the path \texttt{\$.b[*]} are fully included in values at the path \texttt{\$.a[*]}, giving this dependency a strength of 2/3.
Depending on the approximation threshold, we could accept or reject this dependency.
In practice, we would have many more documents and likely set a higher threshold that would reject dependencies of this strength.

\subsection{Functional Dependencies}

The original TANE algorithm included a natural method of incorporating approximate dependencies based on the number of tuples in the input data which violate the dependency.
Using our bitmap formulation, we have no way of easily calculating the number of violating documents.
Instead, we observe conflicting pairs of values.
The same is true in the case of the FDep algorithm.
In the case of FDep, although we express changes in terms of documents, our modifications apply equally to the original relational relational algorithm.
As discussed in Section~\ref{subsec:fdep}, FDep starts by finding the negative cover, that is, a set of dependencies which cannot hold based on conflicting values in pairs of input documents.

In the case of both dynamic TANE and FDep, we can easily obtain a count of the violating pairs of values.
However, setting a threshold for validity based on this metric is unintuitive.
That is, it is difficult to understand how to select a meaningful threshold for this value.
Several other metrics for approximate dependencies have been proposed~\cite{Giannella2004}, but we use
a similar metric to the one proposed in the original TANE implementation, the percentage of violating documents.
This leaves us with the problem of finding the number of individual documents which must be removed for a dependency to hold.
Since we already have pairs of violating documents, we simply need to select one document from each of these pairs, minimizing the total number of documents.
We represent the problem as graph of documents with pairs of violating documents connected by edges.
Then we simply need to find the size of the minimum vertex cover which gives us the smallest number of violating documents.

Since this problem is NP-hard and needs to be solved for every dependency, we opted to resort to a well-known approximate solution.
As we process pairs of violating documents while finding negative dependencies we greedily count both documents in the pair unless one of these documents has been previously counted.
We use a threshold on the percentage of violating documents when selecting valid negative dependencies.
This threshold is approximate since our solution to the minimum vertex cover is also approximate, but we can tune the threshold to meet application requirements based on the approximation.
The final result of approximation for both of these algorithm is a set of dependencies which are valid when a subset of documents are removed from the input.

Again we consider the documents in Figure~\ref{fig:approx_data}.
For the dependency $\texttt{\$.a[*]}\rightarrow\texttt{\$.b[*]}$, we note that all four documents intersect on values for the key \texttt{a}.
When comparing pairs of documents for violations, we would find that documents 2 and 3 as well as 3 and 4 do \emph{not} intersect on values for the key \texttt{b} which would be required by this dependency.
Note that only one of these four documents (document 3) would need to be removed from this collection to make the dependency valid.
However, we are approximating the number of violating documents as outlined above.
Assume we first encounter the violating pair 2 and 3.
We would record that these two documents violate the dependency.
Then when encountering the violating pair 3 and 4, we would see that document 3 has already been recorded, and avoid recording document 4 as a violation.
Finally, the strength of this dependency would be calculated as 2/4 since our approximation has identified two of the four documents as violations.

\section{Evaluation}\label{sec:eval}

Our goal with this section is to compare runtime efficiency of static versus dynamic unrolling for each of our mining algorithms.
We evaluate each algorithms against datasets with varying sizes and levels of nesting to demonstrate the scalability of our dynamic unrolling approach.

All experiments are performed ver 411, 16GB RAM, with data stored on a 256GB Samsung SSD 860 Pro.
Algorithms are implemented in Python 3 and run with an approximation threshold of 99\%.
We limit the runtime of each algorithm to 12 hours.
We use four datasets in our evaluation.
First, we collected data on all 45,000+ cards from the trading card game Magic: The Gathering (MTG) via their public API.
Second, we use a sample of event information collected from the GitHub API.
Third, we use a Amazon product information with a structure similar to the data in Figure~\ref{fig:amz_data}.

To more deeply explore the performance of our algorithms with respect to document complexity, we consider relevant properties of each document.
Since the goal of dependency mining is to collect information about attribute values, the first relevant property we consider is the average number of attribute values per document.
Attribute values we consider are at the deepest level of nesting since those are the ones considered by our dependency mining algorithm.
For example, consider the document below:

\begin{lstlisting}[language=json]
{ "id": 17,
  "authors": [
    {"id": 24, "name": "Maya Angelou"},
    {"id": 56, "name": "Oprah Winfrey"}
  ],
  languages: ["English", "Spanish"] }
\end{lstlisting}

This document has a total of seven attribute values.
One for \texttt{id}, two for each of the two attributes in the \texttt{authors} array, and two for each value of \texttt{languages}.
In addition to the number of attribute values, we also want to consider the level of nesting of these values.
We simply define the level of nesting as the number of objects or arrays that the attribute value is contained in.
For the example above, the value \texttt{17} and \texttt{"English"} is nested at level 2.
The value \texttt{24} is nested at level 3 since it is contained in the top-level document object, the array \texttt{authors}, and an object under the key \texttt{id}.
A simple characterization of our first three datasets according to these properties is given in Table~\ref{tbl:datasets}.

Finally, we use a sampling of posts made to the Reddit social news site.
In the case of the final dataset, we scale the time range under consideration to provide a useful method of testing the scalability of our algorithms.
That is, as we describe later, we consider a variable number of documents in time order.


\subsection{Inclusion Dependency Mining}

We start with an analysis of our dynamically unrolled inclusion dependency mining algorithm.
Our goal is to compare the runtime efficiency of our algorithm on various real-world datasets which we do in Section~\ref{subsubsec:ind-efficiency}.
Furthermore, in Section~\ref{subsubsec:ind-scalability}, we show that dynamic unrolling also scales better to large datasets.

\subsubsection{Efficiency}\label{subsubsec:ind-efficiency}

To demonstrate the effectiveness of our approach of dynamic unrolling, we examined the runtime of dynamic vs. static unrolling on the first three of our datasets.
We ran the each dependency mining algorithm using our modifications of approximations and both static and dynamic unrolling.
The average of three runs for each each dataset, algorithm, and unrolling strategy is given in Table~\ref{tbl:ind_runtime}.

In all cases, dynamic unrolling is more effective than static unrolling.
We see that on most datasets, the DeMarchi algorithm performs significantly better.
However, in the presence of a large arrays as is the case with the Amazon dataset, performance becomes significantly worse.
This is because when multiple values are observed for a path, we need to look up the previously recorded set of values each time to add to the set of these values.
It is possible that further optimizations to our implementation could remove some of this overhead.

\begin{table*}[ht]
\centering
\rowcolors{2}{white}{gray!25}
\begin{tabular}{ccccc}
\textbf{Dataset} & \textbf{Document count} & \textbf{Average size (bytes)} & \textbf{Maximum nesting} & \textbf{Expansion factor} \\
MTG & 45,846    & 2,878 & 2 & 1,481.2$\times$ \\ 
GitHub & 218,939   & 2,292 & 6 & 4.7$\times$ \\ 
Amazon & 9,430,088 & 1,118 & 2 & 5,043.2$\times$ 
\end{tabular}
\caption{JSON dataset descriptions}\label{tbl:datasets}
\end{table*}

\begin{table*}
\centering
\rowcolors{2}{gray!25}{white}
\begin{tabular}{c|ccc|ccc}
\rowcolor{white}
& \multicolumn{3}{c}{\textbf{SPIDER}} & \multicolumn{3}{c}{\textbf{DeMarchi}} \\
& Static & Dynamic & Improvement & Static & Dynamic & Improvement \\
\textbf{MTG} & 716.32 & 19.04 & 37.6$\times$ & 8,828.54 & 42.04 & 210.1$\times$ \\
\textbf{GitHub} & 670.02 & 638.50 & 1.05$\times$ & 242.4 & 223.67 & 1.07$\times$ \\
\textbf{Amazon} & 264,918.68 & 4,289.84 & 61.8$\times$ & >12h & 13,388.16 & -- \\
\end{tabular}
\caption{Inclusion dependency mining runtime in seconds for various datasets}\label{tbl:ind_runtime}
\end{table*}

\subsubsection{Scalability}\label{subsubsec:ind-scalability}

We compare the scalability of our solution of dynamic unrolling on two dimensions.
Firstly, how does the approach scale with an increasing number of documents?
Second, for a fixed collection of documents, how does each approx scale with increasing document complexity?
We answer the first question with a real-world data set and the second by characterizing the complexity of documents and synthetically generating documents with varying degrees of complexity.

To determine how well dynamic unrolling scales with large collections of documents, we make use of a large collection of documents gathered from the Reddit social news service.
We vary the number of documents used as input to the algorithm from 10-200 thousand and measure the average runtime across three executions of the algorithm.
Results are shown in Figure~\ref{fig:ind_scalability} (note that the y-axis is logarithmic).
We can see that dynamic unrolling scales significantly better than static unrolling.
When operating on 200 thousand documents, dependency mining using dynamic unrolling completes over 97$\times$ faster than when using static unrolling.

This is expected since dynamic unrolling processes much less data.
We can see a large increase in runtime at 50 thousand documents when using static unrolling.
This is a result of an increase in document complexity among these documents.
When documents are more complex, static unrolling generates a larger number of records which must be processed.
We refer to the ratio of unrolled records to the original number of documents as the \emph{expansion factor}.
This expansion is not a concern using our dynamic unrolling approach since each attribute requires a single processing action unrelated to the nesting level of the attribute.

We also note that in the dynamic case, the SPIDER and DeMarchi algorithms have very similar performance, while in the case of static unrolling, the DeMarchi algorithm performs significantly worse than SPIDER.
This is because in the presence of repeated values, the SPIDER algorithm simply records duplicates and continues.
Additional disk space is used for these values, but there is minimal impact on the runtime.
In the case of the DeMarchi algorithm, any duplicate value means we need to retrieve and update the previously stored set of paths for this value.
A sample of the Reddit data suggests the expansion factor is approximately 26$\times$, resulting in significant repeated values which will incur this overhead.






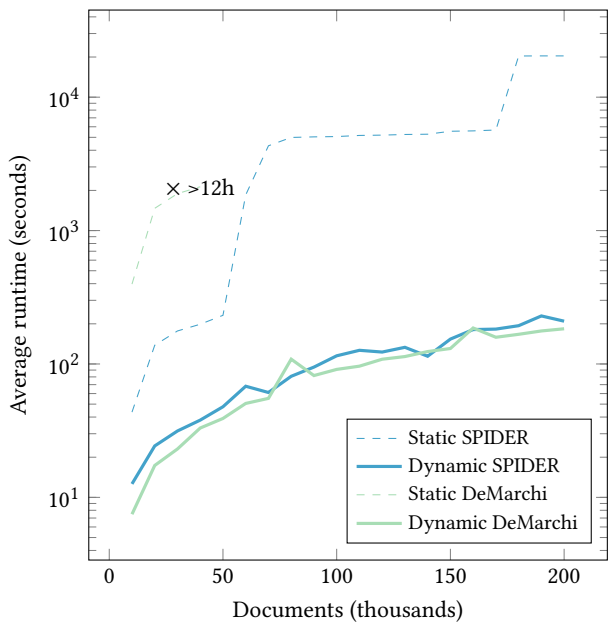
\begin{figure}[ht]
\begin{tikzpicture}
  \begin{semilogyaxis}[
    legend style={font=\small},
    legend cell align={left},
    legend pos=south east,
    xlabel={Documents (thousands)},
    ylabel={Average runtime (seconds)},
    width=\linewidth,
    height=1.05\linewidth,
  ]
    \addplot[darkblue,dashed] table {scale_spider_static.dat};
    \addplot[darkblue,very thick] table {scale_spider_dynamic.dat};
    \addplot[darkgreen,dashed] table {scale_demarchi_static.dat};
    \addplot[darkgreen,very thick] table {scale_demarchi_dynamic.dat};
    \legend{Static SPIDER,Dynamic SPIDER, Static DeMarchi, Dynamic DeMarchi}
    \node at (axis cs:40,2130) {$\times$ >12h};
  \end{semilogyaxis}
\end{tikzpicture}
\caption{Scalability of static vs. dynamic unrolling for inclusion dependency mining}\label{fig:ind_scalability}
\end{figure}

\subsubsection{Examples}

Below we include examples of nested inclusion dependencies which involve nested data and thus have no direct equivalent in the classical definition of inclusion dependencies.

\smallskip
$\texttt{\$.printings[*]}\subset\texttt{\$.set}$
\smallskip

This dependency is on the Magic: The Gathering dataset and indicates that elements of the array at the \texttt{printings} key hold values which also appear as values in the \texttt{set} key.
This is logical since printings refer to all the sets a card was printed in.
We further note that this dependency has a strength of $\sim$99.5\% since there exists a printing which was not the primary set of any card.
This demonstrates the utility of our approach to approximation since the dependency expresses a meaningful semantic relationship which does not hold on the entire dataset.

\smallskip
$\texttt{\$.payload.release.assets[*].uploader.url}$
\\$\subset\texttt{\$.payload.release.author.url}$
\smallskip

On the GitHub dataset, this expressed the dependency that all uploaders of assets for a release are also the authors of the release.
This is expected since GitHub's release API has the release author upload all assets.

\subsection{Functional Dependencies}

Unfortunately, our current approaches to nested functional dependency mining do not scale as well our inclusion dependency mining algorithms.
We were not able to mine a complete set of dependencies from our original datasets in our prescribed time limit.
However, we do demonstrate in this section, that dynamic unrolling still has a significant advantage over static unrolling.
We note that in this case, using static unrolling does \emph{not} produce the same set of dependencies.
This is because the expansion of rows in static unrolling loses the association between values in the original rows.
In the case where there is no nesting however, both static and dynamic unrolling will produce the same results.

Although there is no correct implementation of nested functional dependency mining with static unrolling, we use an implementation of TANE and FDep with static unrolling to demonstrate the savings from processing a smaller dataset.
Since none of the algorithms was unable to complete on our first three datasets, we instead proceed with a smaller version of our scalability experiments with the Reddit dataset.
These results are shown in Figure~\ref{fig:fd_scalability}.

Neither of the algorithms were able to complete on even the smallest dataset of 25 documents with static unrolling.
This is due to the number of expanded records and the poor scalability of our algorithms.
The FDep algorithm ran out of memory on our dataset of 50 documents after running for nearly 40 minutes.
The dynamic version of the TANE algorithm was the only one able to complete successfully on larger collections of documents.
Even so, we only evaluated the performance up to 500 documents which took over 3 hours.
Clearly there is significant work to do in improving the scalability of any of these algorithms.
In the case of TANE which performs the best, the construction of the initial set of bitmaps is still $\mathcal{O}(n^2)$ where $n$ is the number of documents.
We leave the adaptation of other functional dependency mining algorithms and the development of new algorithms as future work.

Although these algorithms scale poorly with an increased number of documents, one possible approach is to use sampling techniques to produce dependencies which are likely to hold.
When disabling any approximation, each algorithm never produces false negatives, but may produce false positives.
That is, it is possible a document which violates a dependency was not present in the sample and the dependency is erroneously reported as valid.
However, all valid dependencies will always be discovered on any sample.
This means we can take several samples of data and find the intersection of all discovered dependencies as an alternative approximation.
While this observation may be of some practical use, further analysis is necessary to provide meaningful estimates on the likelihood each dependency is valid when mined according to this approach.



\begin{figure}
\begin{tikzpicture}
  \begin{semilogyaxis}[
    legend style={font=\small},
    legend cell align={left},
    legend pos=south east,
    xlabel={Documents},
    ylabel={Average runtime (seconds)},
  ]
    \addplot[darkblue,very thick] table {scale_tane_dynamic.dat};
    \addplot[darkgreen,very thick] table {scale_fdep_dynamic.dat};
    \legend{Dynamic TANE, Dynamic FDep}
    \node at (axis cs:50,2600) {OOM};
  \end{semilogyaxis}
\end{tikzpicture}
\caption{Scalability of dynamic unrolling for functional dependency mining}\label{fig:fd_scalability}
\end{figure}
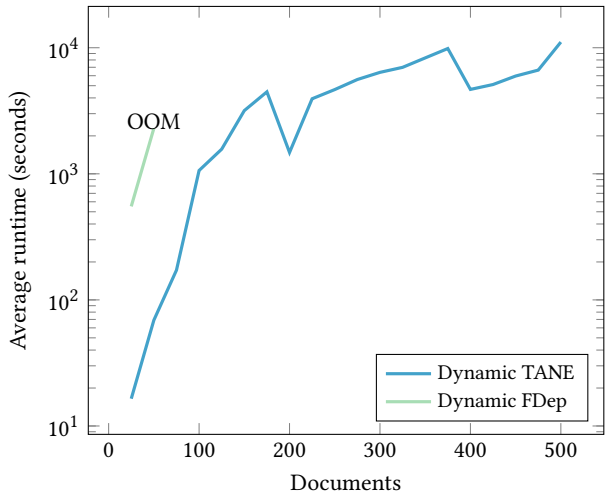

\subsubsection{Examples}

Although our mining algorithms were not able to complete on the entirety of our datasets, we include here some meaningful examples mined from a subset of the data.

\smallskip
$\texttt{\$.colors[*]}\rightarrow\texttt{\$.colorIdentity[*]}$
\smallskip

This dependency holds on our dataset of Magic: The Gathering cards.
It states that if two cards share common colors, they must share common color identities.
According to the definition of color identity, the color of the card is always included in the color identity, so the dependency expresses a meaningful relationship between these two arrays.

\smallskip
$\texttt{\$.brand}\rightarrow\texttt{\$.categories[*][*]}$
\smallskip

The dependency above represents a relationship between brands and product categories on our amazon dataset.
It indicates that whenever two products have the same brand, they should also have some category in common.
While this dependency may not hold across all products, it does represent meaningful semantic information which cannot be expressed with traditional functional dependencies.

\section{Related Work}~\label{sec:related_work}

We split our discussion of related work into two sections.
Section~\ref{subsec:rel_nested} discusses existing definitions of nested dependencies and Section~\ref{subsec:rel_mining} compares our approach with existing algorithms for mining unnested dependencies.
To our knowledge, there are no existing algorithms for mining functional and inclusion dependencies on nested data.


\subsection{Nested Dependencies}\label{subsec:rel_nested}

Hara and Davidson~\cite{Hara1999} also defined a class of nested functional dependencies (NFDs) but do not define any algorithms for mining these dependencies.
These NFDs are functional dependencies defined on a nested relational model.
Their definition differs from ours in one key ways.
Dependencies can be ``local'' to a sub-element of the model (the equivalent of a nested object in our JSON model).
In contrast, our NFDs always involve all values in a document which exist at a particular path.
This locality could be a useful future extension.

Fischer et al.~\cite{Fischer1985} also examine dependencies in the nested relational model.
Their analysis of functional dependencies explores a similar locality property to those proposed by Hara and Davidson.
They also provide analysis of multi-valued dependencies which may be interesting class of dependencies to explore in future work.

\subsection{Dependency Mining}\label{subsec:rel_mining}

Although there has been some previous work in defining dependencies on nested data models, we are not aware of existing work in algorithms to mine dependencies on nested data.
Instead, this section focuses on existing algorithms for non-nested dependency mining.

Papenbrock et al.~\cite{Papenbrock2015-3} analyze seven functional dependency mining algorithms and classified them into three categories: (1) lattice traversal, (2) difference- and agree-set and (3) dependency induction.
The first one comprised of TANE~\cite{Huhtala1999}, FUN, FD\_Mine~\cite{Yao2002}, and DFD~\cite{Abedjan2014}.
The second group is made up of Dep-Miner and FastFDs~\cite{Wyss2001}.
The last category includes FDep~\cite{Flach1999}.
To compare them, the authors evaluated their performance in terms of scalability as the number of rows and columns increase.
The most scalable algorithm varied primarily with the number of columns with TANE and FDep performing the best.
This informed our choice of algorithms to adapt.

As we noted previously, the original implementation of TANE also included a mechanism for approximate mining of functional dependencies which inspired our approach.
De and Kambhampati~\cite{De2010} present a set of algorithms to mine functional dependencies in probabilistic databases.
Our work does not deal with probabilistic databases, although some of the mining techniques proposed may also apply to probabilistic databases with non-relational data models. Note that as far as we are aware, such databases do not yet exist.

Similarly to the study on functional dependencies above, D{\"u}rsch et al.~\cite{Dursch2019} evaluated thirteen algorithms for mining inclusion dependencies.
Our current work focuses on unary inclusion dependency discovery although we would consider n-ary dependency discovery as useful future work.
D{\"u}rsch et al. examine nine algorithms which attempt to mine unary inclusion dependencies.
SPIDER~\cite{Bauckmann2007} which we adapt in this work shows poor scalability with respect to the number of rows.
However, it scales well with the number of columns and (in our experience) is simple to implement.
In the future, we would like to explore S-INDD and S-INDD++ which present improvements over SPIDER.
Furthermore, the performance analysis by D{\"u}rsch et al. also shows that the mining algorithm presented by De Marchi et al.~\cite{DeMarchi2002} exhibits the best scalability relative to the number of rows and columns. 

Kruse et al.~\cite{Kruse2017} proposed \textsc{Faida} which uses approximation techniques to efficiently mine dependencies while potentially producing false positives.
It is also the only inclusion dependency mining algorithm we are aware of which supports approximation.
However, our technique for mining approximate inclusion dependency mining differs in that dependencies which do not hold are a feature since they represent a well-defined possibility of errors in the input data.
However, given the significant performance improvements seen by \textsc{Faida}, a similar approach may warrant future work.

\section{Conclusions and Future Work}

We presented an extension of inclusion dependencies that is useful for identifying relationships in nested JSON data.
We also developed an algorithm capable of mining this class of dependencies on a collection of JSON documents.
Our evaluation of this algorithm on a variety of datasets showed its effectiveness over existing approaches.

While this algorithm has proved effective, there are several limitations to be addressed in future work.
Firstly, we assume that all documents in the collection have uniform schema.
Also, we currently only consider unary inclusion dependencies.
Adapting algorithms such as BINDER~\cite{Papenbrock2015} could enable extension to n-ary inclusion dependencies.

Furthermore, we hope to make use of this algorithm to identify new relationships between large open data sets.
Our current implementation requires any the dataset being examined to fit in memory. 
This is not an inherent restriction to our techniques and we expect this limitation to be lifted in the future.
The ability to perform mining on large datasets is important as we wish to be able to mine dependencies across many collections of data to help determine relationships between these collections.

In previous work, we proposed a three-step process to renormalize a non-relational database~\cite{Mior2018}.
This work required a flattened relational schema to be used as input and dependencies were mined using the equivalent of static unrolling.
We expect this work to be useful in the context of non-relational database renormalization to replace the need to use a relational model as algorithm input.

\bibliographystyle{abbrv}
\bibliography{main}

\begin{appendix}

\section{Alloy model for NFDs}\label{sec:alloy}

Below is a copy of the Alloy model used while developing proofs for the axioms on nested functional dependencies in Section~\ref{subsec:ind_rules}.
We note that this model is incomplete since it does not fully capture the association between keys and values in functional dependencies.
However, we did find this model useful as an aid during proof construction.
It is also useful for constructing particular examples of documents or dependency structures.
We simply ask the Alloy analyzer to assert that the desired properties \emph{never} hold and the counterexample produced is an example which meets our desired properties.

\begin{lstlisting}[language=alloy]
sig Key {} sig Value {}
sig ValueSet { values: set Value }
sig Document {	 data: Key one-> one ValueSet }

// There are some documents with data
fact { #Document > 0 }
fact { all d: Document | #d.data > 0 }

// Keys have at least one value, some more than one
fact { all vs: ValueSet | #vs.values > 0 }
fact { some vs: ValueSet | #vs.values > 1 }

// Value sets are unique
fact { all vs1, vs2: ValueSet |
  vs1 = vs2 || vs1.values != vs2.values }

// All keys exist in some document
fact { all k: one Key | all d: one Document |
  k in d.data.univ }

// All value sets are in some document
// and all values are in some value set
fact { all vs: one ValueSet | some d: one Document |
  vs in d.data[univ] }
fact { all v: one Value | some vs: one ValueSet |
  v in vs.values }

// There are some documents which overlap
// and some which do not overlap
fact { some d1, d2: one Document | all k: one Key |
  d1 != d2 && some d1.data[k].values & 
  d2.data[k].values }
fact { some d1, d2: one Document | all k: one Key |
  d1 != d2 && no d1.data[k].values & 
  d2.data[k].values }

// Get all values for a set of keys
fun values [keys: set Key, d: one Document]:
  set ValueSet { (keys <: d.data)[univ] }

// Some documents must overlap on a set of keys
pred overlap[keys: set Key] {
  some d1, d2: one Document
  // two distinct documents
  | d1 != d2
  // must overlap on all keys in the set
    && some (keys <: d1.data & keys <: d2.data)
}

// A functional dependency must hold
pred fd [lhs: set Key, rhs: set Key] {
  all d1, d2: one Document
  // either the LHS does not overlap
  | (some k: Key | k in lhs && no values[k, d1].values
     & values[k, d2].values)
  // or the RHS does
    || (all k: Key | k not in rhs ||
        some values[k, d1].values
        & values[k, d2].values)
}

check transitivity { all k1, k2, k3: set Key |
  fd[k1, k2] && fd[k2, k3] => fd[k1, k3] }
check reflexivity { all k1, k2: set Key |
  (#k1 > 0 && k1 in k2) => fd[k2, k1] }
check augmentation { all k1, k2, k3: set Key |
  fd[k1, k2] => fd[k1 + k3, k2 + k3] }

\end{lstlisting}
\end{appendix}

\end{document}